\documentclass{article}

\usepackage{arxiv}
\pdfoutput=1
\usepackage[utf8]{inputenc} 
\usepackage[T1]{fontenc}    
\usepackage{hyperref}       
\usepackage{url}            
\usepackage{booktabs}       
\usepackage{amsfonts}       
\usepackage{nicefrac}       
\usepackage{microtype}      
\usepackage{lipsum}
\usepackage{graphicx}
\graphicspath{ {./images/} }

\usepackage{amsthm}
\newtheorem{theorem}{Theorem}
\newtheorem{lemma}{Lemma}
\newtheorem{definition}{Definition}
\newtheorem{proposition}{Proposition}

\usepackage{microtype}
\usepackage{graphicx}
\usepackage{booktabs} 
\usepackage{amsfonts,amsmath,amssymb,amsthm,boxedminipage,color,url}
\usepackage[ruled, vlined,linesnumbered]{algorithm2e} 
\usepackage{algpseudocode}
\usepackage[titletoc,title]{appendix}

\usepackage{amsmath}
\usepackage{amsthm}

\DeclareMathOperator*{\argmin}{\arg\min}

\newcommand{\secref}[1]{Section \ref{#1}}
\newcommand{\agtodo}[1]{({\textcolor{red}{#1}})}

\newcommand{\yuda}[1]{{\textcolor{blue}{#1}}}
\newcommand{\sazulay}[1]{\textcolor{cyan}{\bfseries\{SA: #1\}}}
\newcommand{\Xomit}[1]{}
\newcommand\remove[1]{}

\title{Holdout SGD: Byzantine Tolerant Federated Learning}

\author{
 Shahar Azulay \\
  The Blavatnik School of Computer Science \\
  Tel Aviv University, Israel \\
  \texttt{shaharazulay@mail.tau.ac.il} \\
   \And
 Lior Raz \\
  The Blavatnik School of Computer Science\\
  Tel Aviv University, Israel\\
  \texttt{liorraz@mail.tau.ac.il} \\
   \And
 Amir Globerson \\
  The Blavatnik School of Computer Science\\
  Tel Aviv University, Israel\\
   \And
 Tomer Koren \\
  The Blavatnik School of Computer Science\\
  Tel Aviv University, Israel\\
  \texttt{tkoren@tauex.tau.ac.il} \\
   \And
 Yehuda Afek \\
  The Blavatnik School of Computer Science\\
  Tel Aviv University, Israel \\
  \texttt{afek@tauex.tau.ac.il} \\
}

\begin{document}
\maketitle
\begin{abstract}
This work presents a new distributed Byzantine tolerant federated learning algorithm, \textit{HoldOut SGD}, for Stochastic Gradient Descent (SGD) optimization.
\textit{HoldOut SGD} uses the well known machine learning technique of holdout estimation, in a distributed fashion, in order to select parameter updates that are likely to lead to models with low loss values. 
This makes it more effective at discarding Byzantine workers inputs than existing methods that eliminate outliers in the parameter-space of the learned model. 
\textit{HoldOut SGD} first randomly selects a set of workers that use their private data in order to propose gradient updates. 
Next, a voting committee of workers is randomly selected, and each voter uses its private data as holdout data, in order to select the best proposals via a voting scheme. 
We propose two possible mechanisms for the coordination of workers in the distributed computation of \textit{HoldOut SGD}. The first uses a truthful central server and corresponds to the typical setting of current federated learning. The second is fully distributed and requires no central server, paving the way to fully decentralized federated learning. 
The fully distributed version implements \textit{HoldOut SGD} via ideas from the blockchain domain, and specifically the Algorand committee selection and consensus processes.
We provide formal guarantees for the \textit{HoldOut SGD} process in terms of its convergence to the optimal model, and its level of resilience to the fraction of Byzantine workers. The analysis assumes convexity of the loss, but the method is applicable to general learning scenarios, including deep-learning. 
Empirical evaluation shows that \textit{HoldOut SGD} is Byzantine-resilient and efficiently converges to an effectual model for deep-learning tasks, as long as the total number of participating workers is large and the fraction of Byzantine workers is less than half ($<1/3$ for the fully distributed variant). We show trade-offs between the fraction of Byzantine workers, the convergence confidence of the algorithm and the number of participating workers.
\end{abstract}

\section{Introduction}
\label{intro}
Advancements in machine learning have recently introduced the \textit{Federated Learning} (FL) paradigm \cite{McMahan2016CommunicationEfficientLO, Konecn2018FederatedLS, Kairouz2019AdvancesAO}. FL is a distributed learning framework where the data used to train the model is distributed across many user devices (e.g., mobile phones or personal computers). The defining characteristic of FL is that data is not transmitted from the devices. This is motivated by several factors such as privacy concerns, load sharing of computation, and communication efficiency. Rather than copying data, a centralized \textit{parameter server} \cite{Li2014ScalingDM} orchestrates a distributed process by which the workers collaborate and communicate to learn a model \cite{Zinkevich2010ParallelizedSG,Feng2014DistributedRL,Alistarh2018ByzantineSG}. At a high-level, FL progresses in synchronous epochs, each beginning with the server sending the current model parameters to all workers, which then locally compute, based on their private data, a new local gradient update. Each worker then sends its update to the server, which aggregates all updates to compute a new set of model parameters, typically via a Stochastic Gradient Descent (SGD) step. This process is repeated until a termination condition is reached. 

When the learning process is outsourced to a large crowd, as in FL, it becomes vulnerable to malicious or faulty workers that (deliberately or unintentionally) cause the process to converge to an ineffectual model. Several approaches have been suggested to provide Byzantine fault tolerance in FL \cite{Alistarh2018ByzantineSG,Blanchard2017MachineLW}. All the approaches that have been suggested thus far, to the best of our knowledge, work in the high-dimensional parameter-space of the learned joint model. That is, they discard outlier updates that are far from the ``center'' of all proposed updates (e.g., their mean \cite{Blanchard2017MachineLW} or trimmed-mean \cite{Yin2018ByzantineRobustDL}). While this approach is intuitively reasonable, the selection criterion is not directly related to the loss (and in turn, to the accuracy) of the resulting model. Consequently, Byzantine nodes can propose updates that appear valid but result in a poor model after several iterations \cite{Mhamdi2018TheHV}. 
Formally, \cite{Blanchard2017MachineLW} showed that no aggregation rule based on a linear combination of the updates proposed by the workers tolerates a single Byzantine worker.
This notion led to the introduction of many non-linear aggregation rules, aiming to provide the backbone for a Byzantine-resilient distributed SGD \cite{Chen2017DistributedSM, Blanchard2017MachineLW, Yin2018ByzantineRobustDL, Mhamdi2018TheHV, Alistarh2018ByzantineSG, Xie2018GeneralizedBS, Su2016FaultTolerantMO, Feng2014DistributedRL}.
These methods were recently shown to break under crafted perturbations of the parameter updates, leveraging the high-dimensionality of the learning task and the non–convexity of the target loss function.

Inspired by distributed consensus algorithms and Algorand blockchain committee usage \cite{Chen2019AlgorandAS}, here we suggest a new approach we call \textit{HoldOut SGD}, to eliminate the effect of Byzantine updates. Rather than evaluating a gradient based on geometric criteria, we directly choose gradients based on their contribution to optimizing the accuracy of the model. Our approach is based on the well known holdout method in learning (e.g., \cite{Hastie2005TheEO, Kohavi1995ASO}), where a new dataset is used to choose between different candidate models. Here we use holdout in a very different context, for choosing between different gradient updates proposed by different workers. \textit{HoldOut SGD} uses two randomly selected subsets of workers. The first set of workers, called \emph{proposers}, proposes gradient updates similarly to traditional workers in FL. The second set of workers, called the \emph{voting committee}, is responsible for carefully selecting a subset of the proposed updates, based on the data held by the voting committee members and a voting scheme. Thus, at each round of the algorithm, the data used by the committee can be viewed as holdout data for that round.

We demonstrate the quality and robustness of our scheme both theoretically and empirically. First, we provide convergence guarantees and convergence rates for the algorithm. Second, we show that it is Byzantine-resilient, in the sense that the votes cast by honest voting committee members are resilient to the corrupted updates generated by Byzantine workers, resulting in convergence under a Byzantine fraction $f < 0.5$. We study the probability of the voting committee being compromised by a majority of Byzantine nodes, and derive the size of committee needed to ensure the convergence of the algorithm, with high probability.
While for the theoretical proofs we make some relaxing assumptions (most notably convexity) in the empirical evaluation of \textit{HoldOut SGD} we consider a typical deep-learning scenario, and show that it withstands attacks on which previous approaches failed. 

{\bf Paper Organization:} 
The rest of the paper is organized as follows. Section~\ref{background} provides background on the distributed learning setting with Byzantine workers, and on existing Byzantine-resilient algorithms for distributed SGD. 
Section~\ref{our-algorithm} introduces our algorithm, \textit{HoldOut SGD}.
\secref{theroy} provides theoretical claims for \textit{HoldOut SGD}, while \secref{results} describes empirical results comparing \textit{HoldOut SGD} to other Byzantine-resilient methods for distributed learning.
Finally, \secref{conclusions} gives some concluding remarks.

\section{Background}
\label{background}

\subsection{Distributed Learning with SGD}
\label{distributed_sgd_background}
We follow the distributed model described in \cite{Li2014ScalingDM}, which includes a centralized parameter server, $n$ worker nodes, among which a fraction $f \in [0, 1]$ can be Byzantine nodes.
We assume the Byzantine nodes exhibit arbitrary behavior as individuals, and also have the ability to operate in coalition.

Let $X_1,\ldots,X_m$ be a set of $m$ i.i.d.~samples representing the training set, where $m$ is large (e.g., each $X_i$ is a labeled image).
Let $w$ represent the parameters of a machine learning model to be trained on $X$ (e.g., the model maps between an image and its label). Consider a loss function $L(X; w)$ that measures the discrepancy between the model $w$ and the given example $X$.\footnote{For example, if $X$ is a labeled image and $w$ is a classification model then $L$ is the cross-entropy loss of the model on the labeled image.}

The goal of the training process is to find a model $w^*$ that minimizes the training loss $L(X; w)$ defined as:%
\footnote{Of course, the true goal of learning is to find a model that generalizes well; our analysis can be adapted to give similar bounds for the test loss, but for simplicity, here we focus on training optimization.}
$$L(X; w) = \frac{1}{m}\sum_{i=1}^m L(X_i ; w).$$
In the context of Federated Learning this is done via a variant of Stochastic Gradient Descent defined as follows (see Algorithm~\ref{alg:distributed_sgd}): Each worker $i \in \{1, ..., n\}$\ holds $m_i$\ samples ($\sum_{i=1}^{n} m_i = m$), drawn randomly from the dataset. Let $t$ represent the current epoch, $w_t$ be the learned model parameters at epoch $t$, and $d$ be the dimension of the parameter space $w$.

At epoch $t$, the parameter server randomly selects $P_t$, a group of $N$ nodes (where usually $N \ll n$) as participants in the current round. Each selected node $i \in P_t$ updates its model parameters to $w_t$ (received from the parameter server at the end of epoch $t-1$), randomly draws $B_t^i$ which consists of $B$ samples (also called a ``mini-batch'') from its internal dataset and reports the relevant gradient to the parameter server.
Formally, denote the loss function over the selected mini-batch of node $i$ by $L_i(w_t)$ so that:

$$L_i(w_t) = \frac{1}{B} \sum_{j \in B_t^i} L(X_j; w),$$
and $\nabla L_i(w_t)$ is the gradient of $L_i(w_t)$, reported to the parameter server.

The parameter server then runs an aggregation rule $\mathbb{A}$ over all reported gradients and decides on the resulting update it should take. We denote this selected update by $v_t$.
Using a learning rate $\eta_t>0$, the parameter server takes a step in the direction of the selected update $w_{t+1} \leftarrow w_t - \eta_t v_t$ and reports the updated parameters $w_{t+1}$ to all nodes, declaring the end of the training epoch.

\begin{algorithm}[tb]
   \caption{Distributed SGD: code for parameter server and workers}
   \label{alg:distributed_sgd}
    \SetKwProg{DistributedSGD}{Procedure \emph{DistributedSGD}}{}{end}
    \SetKwProg{getUpdate}{Function \emph{getUpdate}}{}{end}
    \DistributedSGD{($T$, $n$, $w_1$): \text{\# performed by the parameter server}}{
       \For{epoch $t \in [T]$} {
       $P_t$ $\leftarrow$ Select $N$ nodes uniformly at random\;
       \For{node i $\in$ $P_t$} {
            $\nabla L_{i}(w_t)$ $\leftarrow$ $i.getUpdate(w_t)$\;
       }
       $v_t \leftarrow \mathbb{A}(\nabla L_{i_1}(w_t), ..., \nabla L_{i_N}(w_t))$\; 
       $w_{t+1} \leftarrow w_t - \eta_t v_t $\;
       }
   }
   \getUpdate{($w_t$): \text{\# performed by the participant worker}}{ 
      Draw $B$ samples from internal dataset uniformly at random\;
      \Return gradient $\nabla L_i(w_t)$ to server\;
   }

\end{algorithm}

It can be seen that if the aggregation rule $\mathbb{A}$ is an average of all reported gradients then Algorithm~\ref{alg:distributed_sgd} is equivalent to the case of centralized SGD \cite{BottouATStochasticGL} with a mini-batch size of $N B$.

\subsection{Existing Byzantine-Resilient Aggregation Rules}
\label{existing-rules}
As mentioned earlier, Distributed SGD can be very sensitive to Byzantine nodes. Several approaches have been suggested to address this difficulty. These differ mostly in the way they perform aggregation over the participant gradient proposals. 
Existing aggregation rules work by removing gradients that are ``far'' from the mean (or other notions of set center), treating them as adversarial outliers. For example, the Krum method \cite{Blanchard2017MachineLW} returns the set of gradients that has the smaller radius (measured in $L_2$ in parameter space). 
Similarly, the \textit{Coordinate-wise Trimmed Mean} method \cite{Yin2018ByzantineRobustDL} uses an aggregation which evaluates a robust mean around the median.
In \cite{Baruch2019ALI} it was shown that these approaches are in fact prone to attacks where small directed changes to many parameters can take advantage of the non-convexity of the loss function, causing the learning process to converge into an ineffectual model.
\section{HoldOut SGD}
\label{our-algorithm}

\subsection{The Algorithm}
An epoch of \textit{HoldOut SGD} begins like regular DistributedSGD, where at each epoch $t$ a set of ``proposer'' nodes $P_t$ are randomly selected, and calculate their gradients based on their private data. 
Next, \textit{HoldOut SGD} randomly selects a subset of the workers as a voting committee, which we refer to as $C_t$.  Each committee member then evaluates the proposed updates on its own private dataset, to obtain a direct estimate of the loss function. 
The intuition behind this approach is that the private data of a (honest) committee member can be used to evaluate the different proposed model updates. In particular, model updates proposed by Byzantine nodes are less likely to work well on the private data of a committee member.

\remove{
We show in Section~\ref{theroy} that with high probability, with a large enough committee, the committee is populated by a majority of honest nodes, 
under the random selection process.
}

The core of the algorithm is the use of the voting committee for generating a model update out of the proposed updates, which we explain next.
Let $N_p$ denote the number of proposer nodes, and $N_c$ the number of voting committee nodes.
%
At round $t$, each proposer $i \in P_t$ reports the gradient $\nabla L_i(w_t)$ to the parameter server using a mini-batch of size $B$, as in the standard Distributed SGD algorithm (see Algorithm~\ref{alg:distributed_sgd}, \secref{distributed_sgd_background}).\footnote{We assume the Byzantine nodes exhibit arbitrary behavior as individuals, and also have the ability to operate in coalition.}
The parameter server then sends the $N_p$ received proposals to each voting committee member. Each voter \remove{ing committee member} $c \in C_t$ then calculates the loss value of each of the $N_p$ updates, on its own private data, namely, it calculates $L_c(w_t - \eta \nabla {L_j}(w_t))$ for $j\in P_t$. 
This set of $N_p$ values is then sorted and the indices of the $N_p(1-f)$ smallest values (best ranked) are returned. In other words, committee member $c$ returns the $N_p(1-f)$ most promising proposals,\remove{as judged according} relative to its private data.  
We use $\mathbb{A}_{Holdout}$ to refer to the function that takes the $N_p$ proposals and returns the $N_p(1-f)$ indices.
\remove{
which in turn runs a holdout evaluation over all proposed gradients based on $m_c$ random samples from its private data.  We refer to the evaluations of each voting committee member as \textit{votes} over the proposed updates, and denote the holdout evaluation function by $\mathbb{A}_{Holdout}$ (see Algorithm~\ref{alg:holdout_sgd}).
\agtodo{changed to $N_p,N_c$ below as needed} To calculate $\mathbb{A}_{Holdout}$ each voting committee member $c$ evaluates $L_c(w_t - \eta \nabla {L_j}(w_t))$, the loss over the random $m_c$ samples internally held by $c$ over each of the reported proposer updates $j \in [N]$. The committee member then votes for the top $N(1 - f)$ proposals. Namely, those who provided an update that reduced the loss function $L_c$ the most. Each voter sends its votes to the central server as an unsorted list of the corresponding worker (proposers) ids.
}

In the final step, the parameter server integrates the information from all committee members. This is done via a procedure $\mathbb{A}_{Consensus}$ that takes as input all $N_c$ votes, and proposed gradients, finds a set of proposals, each of which has\remove{that} received a sufficient number of votes from committee members, and returns their average. See \secref{sec:unioncon} for details.
\remove{
The parameter server then runs a second aggregation rule $\mathbb{A}_{Consensus}$, as described in \secref{sec:unioncon}, and reports the resulting new updated model parameters $w_{t+1}$ to all workers and starts the next epoch.
}
\begin{algorithm}[tb]
   \caption{HoldOut SGD: code for parameter server \& workers}
   \label{alg:holdout_sgd}
   
   \SetKwProg{HoldOutSGD}{Procedure \emph{HoldOutSGD}}{}{end}
   \SetKwProg{getUpdate}{Function \emph{getUpdate}}{}{end}
   \SetKwProg{getVotes}{Function \emph{getVotes}}{}{end}
   \HoldOutSGD{($T$, $n$, $w_1$): \ \ \  \text{\# performed by the parameter server}}{
       \For{epoch t $\in$ [T]}{
        $P_t$ $\leftarrow$ Select $N_p$ nodes uniformly at random\;
        $C_t$ $\leftarrow$ Select $N_c$ nodes uniformly at random\;
       \For{node i $\in$ $P_t$} {
            $\nabla L_i(w_t)$ $\leftarrow$ $i.getUpdate(w_{t})$

       }
       \For{node c $\in$ $C_t$} {
            ${Votes}_c$ $\leftarrow$ $c.getVotes(\nabla L_{i_1}(w_t), ..., \nabla L_{i_{N_p}}(w_t))$
       }
       $v_t \leftarrow \mathbb{A}_{Consensus}(\{{Votes_{c_1}, ..., Votes_{c_{N_c}}}\}, \{\nabla L_{i_1}(w_t), ..., \nabla L_{i_{N_p}}(w_t)\})$\;
       $w_{t+1} \leftarrow w_t - \eta_t v_t$\;
       }
   }
   \getUpdate{($w_{t}$): \ \ \  \text{\# performed by the proposer worker}}{
      Draw $B$ samples uniformly at random from internal dataset\;
      \Return gradient $\nabla L_i(w_t)$ to server\;
   }
   \getVotes{($\nabla L_{i_1}(w_t), ..., \nabla L_{i_{N_p}}(w_t)$): \ \ \ \text{\# performed by the committee member worker}}{
        Draw $m_c$ samples uniformly at random from internal dataset\;
        ${Votes} \leftarrow \mathbb{A}_{Holdout}(\nabla L_{i_1}(w_t), ..., \nabla L_{i_{N_p}}(w_t))$\;
        \Return $Votes$ to server\;
   }
\end{algorithm}

It can be seen that if the aggregation rule $\mathbb{A}_{Holdout}$ is simply a random choice (each committee member selects an update from one proposer at random) and the aggregation rule $\mathbb{A}_{Consensus}$ is an average over all updates that received any votes, then Algorithm~\ref{alg:holdout_sgd} is equivalent to the \remove{case of} centralized SGD with a mini-batch size of $N_c B$.

\subsection{The Union-Consensus}
\label{sec:unioncon}
\remove{
In the complete implementation of \textit{HoldOut SGD} we use $N_c \gg 1$ to ensure with high probability an honest representation in the committee (see Section~\ref{theroy}).\agtodo{this is a bit out of context. You should mention it where we need the majority.}
\remove{Therefore, the parameter server needs to reduce many votes collected from committee members into a single update $v_t$ it eventually takes.}
}
\remove{
Once all the votes are collected by the parameter server, it then runs the aggregation rule $\mathbb{A}_{Consensus}$ to reach a final update for the current epoch.
$\mathbb{A}_{Consensus}$ is constructed as follows. The parameter server accumulates the number of votes received (by different committee members) for each proposer. The server then selects the proposers that received at least $N_c (1 - f)$ votes. We refer to this group as the \textit{Union-Consensus} of round $t$ (or $UC_t$) . The final update is the average of the updates received by all proposers in the \textit{Union-Consensus}.}

To calculate $\mathbb{A}_{Consensus}$ the parameter server selects the proposers that have received at least $N_c (1 - f)$ votes. We refer to this group as the \textit{Union-Consensus} of round $t$ (or $UC_t$). The output of $\mathbb{A}_{Consensus}$ is simply the vector $v_t$ that is the average of vectors in $UC_t$, namely:
\begin{align} \label{eq:nable_L}
 v_t = \frac{1}{|UC_t|}\sum_{i \in UC_t} \nabla L_i(w_t) .
\end{align}
In order to show that this is well defined, we show below that the set $UC_t$ cannot be empty.
\begin{lemma}
The Union-Consensus group cannot be an empty set.
\end{lemma}
\begin{proof}
By contradiction; assume that the \textit{Union-Consensus} is empty. Then no proposer has been voted for by at least $N_c(1 - f)$ committee members. 
Let $|V|$ denote the total number of votes received for all the proposers. Then it follows that $|V| < N_p \cdot N_c(1 - f)$.
However, each of the $N_c$ committee members must cast $N_p(1 -f)$ unique votes, meaning that $|V| = N_c \cdot N_p(1 -f)$, leading to a contradiction. 
\end{proof}
The intuition behind the $\mathbb{A}_{Consensus}$ construction is that, first, the honest voters vote for the top $N_p(1 - f)$ proposers, not vouching for updates that are expected to be least effectual when evaluated over the private data held by each voter.
This approach is similar to the construction of the score in \textit{Krum} \cite{Blanchard2017MachineLW} and in \textit{Trimmed Mean} \cite{Yin2018ByzantineRobustDL}.
Second, in order to reach a consensus, only proposers that receive at least $N_c(1 - f)$ votes are taken into account, making sure that if a Byzantine node is selected it has to ``win the trust'' of $N_c(1 - 2f)$ honest committee members (since it can create a coalition with the rest of the $f$ Byzantine nodes). Notice, that the random selection of a new set of proposers and voting committee at each iteration, prevents the Byzantine nodes\remove{cannot adapt} from trying to resemble to a particular truthful node or set of nodes \remove{  or evolve from round to round}that tend to lead to a less effectual joint model.

As an example, consider the case where the ratio of Byzantine nodes in the population is $f = \frac{1}{3}$. In this case $\frac{2}{3}N_c$ of the committee members are expected to be honest and a Byzantine node would need to convince at least $N_c(1 - 2f) = \frac{1}{3}N_c$ honest committee members, meaning at least half of the honest committee members, that it is better (when evaluated over their true internal data) than other honest proposers.

This structure constrains the Byzantine workers from controlling the proposers or voting committee. Since these groups are chosen at random in each round, planning ahead or controlling the selection process is not possible.

Finally, note that for the case $f = 0$, \textit{HoldOut SGD} reduces to centralized SGD with a mini-batch size of $N_p \cdot B$, since the committee has no effect in this case.

\subsection{A Fully Decentralized Implementation \label{sec:fully}}

In most practical situations a central server exists and the above distributed \textit{HoldOut SGD} algorithm provides a good solution.  However, if workers do not trust the central server, or the server may fail, a fully distributed implementation, without a central parameter server, is then desired. A simple approach is for each worker to perform the central server algorithm locally. Each committee member (proposer, or voter) instead of sending its gradient proposal, or votes, to the central server, broadcasts its message to all workers, thus increasing the communication complexity from $O(N_p +N_c)$ to $O(n N_p + n N_c)$.  

There are still a few difficulties that have to be overcome. First, members of each of the different committees need to be randomly selected without the Byzantine workers being able to affect the selection, or predict the set of workers to be chosen (otherwise they might DDoS them). Secondly, since the Byzantine voters may send different votes to different sets of workers, they may cause the honest consensus workers to compute slightly different sets of model parameters, each being legal (see Algorithm \ref{alg:holdout_sgd_fully_decentralized}). In order to mimic a \textit{HoldOut SGD} central server operations the honest workers should agree on the same set of model parameters at the end of each epoch. These difficulties resemble the blockchain model (in some sense also the State Machine Replication, but SMR is not applicable here). Hence we borrow techniques from the blockchain domain, specifically the Algorand committee selection and consensus algorithm \cite{Chen2019AlgorandAS} to coordinate the operations of the workers in the same way as nodes in Algorand coordinate and synchronize.

Following \cite{Chen2019AlgorandAS}, we assume the setup in which each worker has selected a pair of public and secret cryptography keys, and a SHA256 random oracle hash function has been agreed upon, before the algorithm starts. After the setup has been established and shared by all workers, a procedure to select a random string $S_1$ is invoked, to be used in the committees membership selection process. This and the requirement for a consensus  step, makes the fully distributed algorithm resilient to $f < 1/3$ Byzantine workers rather than $< 1/2$ in the semi-distributed variant.  A pseudo algorithm is provided in Algorithm \ref{alg:holdout_sgd_fully_decentralized}.

\begin{algorithm}[t]
   \caption{Decentralized HoldOut SGD: code for worker $i$, $i \in [n]$}
   \label{alg:holdout_sgd_fully_decentralized}
   
   \SetKwProg{DecentralisedHoldOutSGD}{Procedure \emph{DecentralisedHoldOutSGD}}{}{end}
   \SetKwProg{getUpdate}{Function \emph{getUpdate}}{}{end}
   \DecentralisedHoldOutSGD{($T$, $w_1$):}{
       $w^i_1 \leftarrow w_1$\;
       \For{epoch $t \in [T]$} {  
              \ \ \ \ \ \ \ \ {***--- Proposing round ---  ***}\;
        $proposer$ $\leftarrow$ ($H(S_t,``P",Sign_{sk_i}(t)) < q_1$) \Comment{Using random oracle $H$ (SHA256), $sk_i$ worker $i$'s secret key, $q_1$ probability of being a proposer}\; 
       \If{$proposer$} {
            \textbf{Broadcast} ($getUpdate(w^i_{t})$,$Sign_{sk_i}(t)$)

       }
       \ \ \ \ \ \ \ \ {***--- Voting round ---  ***}\;
        $voter$ $\leftarrow$ ($H(S_t,``V",Sign_{sk_i}(t)) < q_2$)\Comment{ $q_2$ probability of being a voter}\;        
       \If{$voter$}{ 
            Receive $(\{\nabla L_{j_1}(w_t^{j_1}),Sign_{sk_{j_1}}(t)),\ldots, (\nabla L_{j_{N'}}(w_t^{j_{N'}}),Sign_{sk_{j_{N'}}}(t)\})$\;
            Verify received proposed gradients\;
            ${Votes} \leftarrow \mathbb{A}_{Holdout}(\nabla L_{j_1}(w_t^{j_1}), ..., \nabla L_{j_{N'}}(w_t^{j_{N'}}))$\;
            \textbf{Broadcast} ($Votes$, $Sign_{sk_i}(t)$)

       }
           \ \ \ \ \ \ \ \ {***--- Holdout Soft Consensus Convergence round ---  ***}\;
       Receive $\{(Votes_{j_1},Sign_{sk_{j_1}}(t)),\ldots, (Votes_{j_{N_c'}},Sign_{sk_{j_{N_c'}}}(t))\}$\;
       $v_t^i \leftarrow \mathbb{A}_{Consensus}(\{{Votes}_{j_1}, ...,Votes_{j_{ N'_c}}\})$\;
       $w'^i_{t+1} \leftarrow w^i_t - \eta_t \cdot v_t^i$\;
    {***- Consensus round (use Algorand Consensus) agree on common model -***}\;
        $Cons-committee$ $\leftarrow$ ($H(S_t,``C",Sign_{sk_i}(t)) < q_3$)\Comment{ $q_3$ probability of being a committee member}\; 
        \If {$Cons-committee$} {$w^i_{t+1} :=$ Algorand Consensus with $w'^i_{t+1}$ as initial suggestion\;
        Share $w^i_{t+1}$ with all\Comment{incurring additional $nN_{con}$ messages}\;
        }
       }
   }
   \getUpdate{($w^i_{t}$): \ \ \  \text{}}{
      Draw $B$ samples uniformly at random from internal dataset\;
      \Return gradient $\nabla L_i(w^i_t)$ to server\;
   }
\end{algorithm}

\section{HoldOut SGD: Theoretical Claims}
\label{theroy}

In this section we prove theoretical properties of \textit{HoldOut SGD}, establishing its convergence under appropriate assumptions and proving the resilience of the committee to Byzantine inputs. 
Our main result is Theorem \ref{theorem-convergnece}. 

\subsection{Assumptions}
\label{assumptions}
We assume that for all $X$, the loss function $L(X,w)$ is differentiable and $\beta$-smooth with respect to $w \in W$ (in other words, it has a $\beta$-Lipschitz gradient).
Recall that if $f(w)$ is $\beta$-smooth then for all $w,w'$ it holds that:
\begin{align}
    f(w')
    \leq
    f(w) + \nabla f(w) \cdot (w'-w) + \frac{\beta}{2} \|w'-w\|^2 \label{eq:smoothness}
    .
\end{align}
We will also assume that each $L(X,\cdot)$ is $G$-Lipschitz, that is $\|\nabla L(X; w)\| \leq G$ for all $X$ and $w$. 
Thus, the variance of the gradient reported by any honest node is bounded by $\mathbb{E}[\|\nabla L_i(w) - \nabla L(w)\|^2] \leq G^2/B$, where $B$ is the node's mini-batch size.
Finally, we assume that the average loss $L(w)$ is $\alpha$-strongly convex, which in particular implies that for all $w$:
\begin{align}
    \|\nabla L(w)\|^2
    \geq
    2\alpha (L(w)-L(w^*)) \label{eq:LP}
    .
\end{align}
Where $w^*$ is optimal minimizer of $L$.

The above assumptions are standard and common in the analysis of convex optimization algorithms (e.g., see \cite{Hazan2016IntroductionTO} for more background and references).
Our analysis focuses on providing guarantees of convergence to a global minimum in the convex case. Some of our arguments can be adapted to show convergence to a critical point in the non-convex setting.

\subsection{Notations}
 The \textit{HoldOut SGD} algorithm takes as input the parameters $\{N_p, N_c, f, B, m_c\}$, where $f$ reflects the expected proportion of Byzantine workers in the population, $m_c$ is the number of samples used by a committee member to perform the holdout evaluation, and $B$ the mini-batch size used by a proposer to generate the reported gradient. Let $j_{(c)}$ be the index of the proposer $j$ selected by some committee member $c$, let $L_c$ be the loss defined over its internal dataset and $v_t$ the final update used at round $t$ (see Eq.~\ref{eq:nable_L}).
 
\subsection{Byzantine Gradient Tolerance}

We define a condition on any aggregation rule to ensure that it is Byzantine-resilient to corrupted gradients at each iteration $t$. 
We show that \textit{HoldOut SGD} satisfies this condition and in the next section use this result to prove that \textit{HoldOut SGD} is Byzantine-resilient and converges close to the optimal solution even under the presence of Byzantine workers.

As suggested by \cite{Blanchard2017MachineLW}, the aggregation rule should output a vector $v_t$ that is not too far from the
actual gradient $\nabla L(w_t)$.
Since the actual gradient points in the direction of steepest ascent, we would like to place a lower bound over the inner-product between the suggested gradient and the actual one. If this inner-product is bounded from below, it limits the ability of a Byzantine worker to cause the aggregation rule to choose an update in a direction too far away from the actual gradient.

We show that for every honest committee member $c$, the updates voted by $c$ are bounded by their distance from the true gradient of $L_c$ (based on the internal dataset of $c$), and use this result to prove that \textit{HoldOut SGD} is Byzantine-resilient.

\begin{definition}
An aggregation rule $\mathbb{A}$ is $\varepsilon$-Byzantine-gradient-tolerant, if it satisfies that:
$$\mathbb{E}[v_t \cdot \nabla L(w)] \ \  \geq \|\nabla L(w)\|^2 - \varepsilon,$$
where $v_t = \mathbb{A}(\nabla L_1(w), ..., \nabla L_N(w))$.
\end{definition}

\begin{proposition}
\label{proposition-1} Assume that $f<\frac{1}{2}$ and that the proposers have an honest majority. Then
\textit{HoldOut SGD} is $(\frac12 \beta G^2\eta_t)$-Byzantine-gradient-tolerant w.r.t.~to any honest committee member $c$ at any iteration t. 
Namely:
$$
    \mathbb{E}[\nabla L_{j_{(c)}}(w_t) \cdot \nabla L_c(w_t)]
    \geq
    \mathbb{E}[\|\nabla L(w_t)\|^2] - \frac12 \beta G^2\eta_t
    .
$$
\end{proposition}
Proposition~\ref{proposition-1} can be interpreted as follows: at every round $t$, each honest committee member holds a gradient $\nabla L_c(w_t)$, which is an unbiased estimate of the true gradient $\nabla L(w_t)$. The proposition shows that the presence of a proportion of $f < \frac{1}{2}$ Byzantine proposers cannot impact the voting process of an honest committee member $c$ too much, with respect to its unbiased estimated of the gradient $\nabla L_c(w_t)$.

\begin{proof}
Consider an update $j_{(c)}$ voted for by an honest committee member $c$, at iteration t. By convexity and $\beta$-smoothness of $L_c$, we have for any honest proposer $i$ that
\begin{align*}
    0
    \leq
    L_c(w_t &- \eta_t \nabla L_i(w_t)) - (L_c(w_t) - \eta_t \nabla L_i(w_t) \cdot \nabla L_c(w_t))
    \leq
    \frac12 \beta\eta_t^2 \|\nabla L_i(w_t)\|^2
    ,
\end{align*}
and since $\|\nabla L_i(w_t)\| \leq G$ this implies
\begin{align*}
    \big| L_c(w_t &- \eta_t \nabla L_i(w_t)) - L_c(w_t) + \eta_t \nabla L_i(w_t) \cdot \nabla L_c(w_t) \big|
    \leq
    \frac12 \beta G^2\eta_t^2
    .
\end{align*}
Thus, minimizing $L_c(w_t - \eta_t\nabla L_i(w_t)) - L_c(w_t)$ is equivalent to maximizing the inner-product $\nabla L_i(w_t) \cdot \nabla L_c(w_t)$ up to an additive $\beta G^2\eta_t/2$.

From the definition of $\mathbb{A}_{Holdout}$ it holds for $f < \frac12 $ and an honest majority of proposers, that there exists at least one honest proposer $i$ such that: $$L_c(w_t - \eta_t \nabla L_{j_{(c)}}(w_t)) \leq  L_c(w_t - \eta_t \nabla L_{i}(w_t)).$$
For the above honest proposer $i$ we have:
\begin{align*}
    \nabla L_{j_{(c)}}(w_t) \cdot \nabla L_c(w_t)
    \geq
    \nabla L_i(w_t) \cdot \nabla L_c(w_t) - \frac12 \beta G^2\eta_t
    .
\end{align*}
The proposition follows by taking the expectation of this inequality and noting that $L_i$ and $L_c$ are independent and $\mathbb{E}[\nabla L_i(w_t)] = \mathbb{E}[\nabla L_c(w_t)] = \nabla L(w_t)$.
\end{proof}

\subsection{Honest Majority \label{sec:honest_majority}}
The condition that both voters and proposers sets have a majority of honest nodes is at the core of the \textit{HoldOut SGD} algorithm.
Both Proposition~\ref{proposition-1} and Theorem~\ref{theorem-convergnece} rely on this assumption.
However, since both committee and proposers groups are randomly selected at each training epoch, this condition could fail with a non-zero probability. In what follows we show that this failure probability can be made arbitrarily low via an appropriate choice of committee size.

Given a confidence level $\delta>0$, let:
\begin{equation}
N(T, \delta) =  2 \frac{(1 + 2f)}{(1 - 2f)^2} \ln{\frac{T}{\delta}}.
\end{equation}
The next lemma shows that if the size of the committee at least $N(T, \delta)$, the committee is guaranteed to have a majority of non-Byzantine nodes with probability greater than $1-\delta$.


\begin{lemma}
\label{lemma_honest_majority}
Let $n$ be a set of nodes, and $f < \frac{1}{2}$ be the fraction of Byzantine nodes in $n$.
Let $A_t$ be group of N nodes, selected uniformly at random without replacement from the $n$ nodes at iteration t. Then if $N\geq N(T, \delta)$, the probability that one of the sets $A_1,\ldots,A_T$ has a majority of Byzantine nodes is smaller than $\delta$.  
\end{lemma}

As can be observed, the derived lower bound grows logarithmically with $T$, the number of iterations, and $\frac{1}{\delta}$, representing the confidence level.
The lower bound also depends on $f$, the proportion of Byzantine nodes in the population, and grows as $f$ grows closer to $\frac{1}{2}$.

\begin{proof}
Let $X_t$ represent the number of selected Byzantine nodes in $A_t$.
Then 
$X_t$ follows a ${Hypergeometric}(n, n f, N)$ distribution.
From the properties of the Hypergeometric distribution we get:
$$\mu = \Bar{X_t} = N \cdot \frac{n \cdot f}{n} = N \cdot f.$$
From Chernoff bounds we have that for all $\epsilon>0$ it holds that
$P(X_t \geq (1 + \epsilon)\mu) \leq e^{-\frac{\epsilon^2\mu}{2 + \epsilon}}.$
The probability of a Byzantine majority at iteration $t$ can then be bounded as follows:
$$P\left(X_t \geq \frac{N}{2}\right) = P\left(X \leq \left(1 + \Big(\frac{1}{2f} - 1\Big)\right)\mu\right) \leq e^{-\frac{(1 - 2f)^2}{(1 + 2f)}\frac{N}{2}},$$
where $\epsilon = \frac{1}{2f} - 1 \geq 0$ because $f<\frac12$. To upper bound the probability that one of the sets has a Byzantine majority we use the union bound:
$$P\left(\bigcup\limits_{t=1}^{T} \Big\{X_t \geq \frac{N}{2}\Big\}\right) \leq T \cdot e^{-\frac{(1 - 2f)^2}{(1 + 2f)}\frac{N}{2}} \leq \delta,$$
where the last inequality follows from  $N\geq N(T, \delta)$.
\end{proof}

\subsection{Convergence Analysis}
\label{convergence-analysis}

In this section we analyze the convergence rate of the \textit{HoldOut SGD} algorithm and prove that it is Byzantine-resilient. 
We assume that the loss function $L(X,\cdot)$ is $G$-Lipschitz and $\beta$-smooth for any $X$ and that the expected risk $L$ is $\alpha$-strongly convex. 
For the SGD step-size we take $\eta_t = 1/(2\alpha t)$.

Our main theorem below states \textit{HoldOut SGD} converges at a rate of $O(\frac{\log T}{T})$ to a ball of radius $O(\frac{1}{\sqrt{m_c}})$ around the optimal $w^*$.
Importantly, the theorem holds whenever the fraction of Byzantine workers is $f<0.5$. 


\begin{theorem}
\label{theorem-convergnece}
Let $f < \frac12$ be the fraction of Byzantine nodes and $\delta>0$ a desired confidence level. Then the following holds: if $N_c, N_p \geq N(2T, \delta)$ then with probability greater than $1 - \delta$ the error of \textit{HoldOut SGD} satisfies: 
\begin{align*}
    \mathbb{E}[L(w_T)] - L(w^*)
    =
    O\left( \frac{G^2}{\alpha\sqrt{m_c}} + \frac{\beta G^2}{\alpha^2} \frac{\log{T}}{T} \right)
    ,
\end{align*}
where $w^*$ is the minimizer of $L$.
\end{theorem}

\begin{proof}
Let us first observe the step taken by a single update $j_{(c)}$ voted for by some honest committee member $c$.
Since $L$ has Lipschitz first derivative we know that:
\begin{align*}
L(w_t &- \eta_t \nabla L_{j_{(c)}})  \leq L(w_t) -\eta_t \nabla L_{j(c)}(w_t)\cdot \nabla L(w_t) + \frac12 \beta\eta_t^2 \|\nabla L_{j(c)}\|^2.
\end{align*}
Writing this differently:
\begin{align*}
&L(w_t - \eta_t \nabla L_{j_{(c)}})
\\
&\leq L(w_t) -\eta_t \nabla L_{j(c)} \cdot (\nabla L(w_t) - \nabla L_c(w_t) + \nabla L_c(w_t)) + \frac12 \beta\eta_t^2 \|\nabla L_{j(c)}\|^2 \\
&= L(w_t) -\eta_t \nabla L_{j(c)} \cdot (\nabla L(w_t) - \nabla L_c(w_t))  -\eta_t \nabla L_{j(c)} \cdot \nabla L_c(w_t) + \frac12 \beta\eta_t^2 \|\nabla L_{j(c)}\|^2.
\end{align*}
Taking the expectation conditioned on all randomness before iteration $t$ (and using $\mathbb{E}_t[\cdot]$ to denote this), and using the bound over the stochastic gradients, we get:
\begin{align*}
&\mathbb{E}_t[L(w_t - \eta_t \nabla L_{j_{(c)}})]
\\
&\leq L(w_t) + \eta_t \mathbb{E}_t[\nabla L_{j(c)} \cdot (\nabla L_c(w_t) - \nabla L(w_t))] - \eta_t  \mathbb{E}_t[\nabla L_{j(c)} \cdot \nabla L_c(w_t)] + \frac12 \beta\eta_t^2G^2
.
\end{align*}
%
%
Using Proposition~\ref{proposition-1}:\footnote{Note that the conditions of \ref{proposition-1} hold here with probability greater than $1-\delta$ because $N\geq N(\delta)$ and Lemma~\ref{lemma_honest_majority}.} 
\begin{align}\label{eq1}
&\mathbb{E}_t[L(w_t - \eta_t \nabla L_{j_{(c)}})]
\\
&\leq L(w_t) + \eta_t \mathbb{E}_t[\nabla L_{j(c)} \cdot (\nabla L_c(w_t) - \nabla L(w_t))] - \eta_t  \|\nabla L(w_t)\|^2 + \frac12 \beta\eta_t^2G^2 + \frac12 \beta\eta_t^2G^2 
\nonumber\\
&\leq L(w_t) + \eta_t \sqrt{\mathbb{E}_t[\|\nabla L_{j(c)}\|^2] \, \mathbb{E}_t[\|\nabla L_c(w_t) - \nabla L(w_t)\|^2]} - \eta_t  \|\nabla L(w_t)\|^2 + \frac12 \beta\eta_t^2G^2 \nonumber\\
&\leq L(w_t) + \frac{\eta_tG^2}{\sqrt{m_c}} - \eta_t  \|\nabla L(w_t)\|^2 + \beta\eta_t^2G^2
,
\end{align}
where we use the Cauchy-Schwartz inequality and the bounds over the stochastic gradients in the last two transitions.
Using the definition of \textit{HoldOut SGD} we can write:
$$w_{t+1} = w_t - \eta_t \cdot v_t = w_t - \eta_t \frac{1}{|UC_t|}\sum_{i \in UC_t} \nabla L_i(w_t) =  \frac{1}{|UC_t|}\sum_{i \in UC_t} (w_t - \eta_t \nabla L_i(w_t)).$$
Since $L$ is convex:
$$L(w_{t+1}) \leq \frac{1}{|UC_t|}\sum_{i \in UC_t} L(w_t - \eta_t \nabla L_i(w_t)).$$
Using Lemma~\ref{lemma_honest_majority} with $N(2T, \delta)$ we get that with a confidence level of $\delta$ both proposers and committee groups are populated with an honest majority of members. 
Therefore, each $\nabla L_i(w_t), i \in UC_t$ was voted by at least one honest committee member $c$ and can therefore be written as $\nabla L_{i(c)}(w_t)$.
Using Eq.~\ref{eq1} we can write:
\begin{align*}
\mathbb{E}_t[L(w_{t+1})] &\leq \frac{1}{|UC_t|}\sum_{i \in UC_t} L(w_t - \eta_t \nabla L_{i(c)}(w_t)) \\
&\leq \frac{1}{|UC_t|}\sum_{i \in UC_t} \big{[}L(w_t) +  \frac{\eta_tG^2}{\sqrt{m_c}} - \eta_t  \|\nabla L(w_t)\|^2 + \beta\eta_t^2G^2\big{]}\\
& = L(w_t) + \frac{\eta_tG^2}{\sqrt{m_c}} - \eta_t  \|\nabla L(w_t)\|^2 + \beta\eta_t^2G^2
.
\end{align*}
Subtracting $L(w^*)$ from both sides and taking the expectation before iteration t, we have:
\begin{align*}
\mathbb{E}[L(w_{t+1}) - L(w^*)] &\leq  \mathbb{E}[L(w_{t}) - L(w^*)] - \mathbb{E}[\eta_t  \|\nabla L(w_t)\|^2] + \beta\eta_t^2 G^2 + \frac{\eta_t G^2}{\sqrt{m_c}}
.
\end{align*}
Since $L$ is $\alpha$-strongly convex, it respects Eq.~\ref{eq:LP} in \secref{assumptions}:
\begin{align*}
\mathbb{E}[L(w_{t+1}) - L(w^*)] &\leq (1 - 2\alpha\eta_t) \mathbb{E}[L(w_{t}) - L(w^*)] + \beta\eta_t^2 G^2 + \frac{\eta_t G^2}{\sqrt{m_c}}
.
\end{align*}
Unfolding the recursion, we obtain
\begin{align*}
    &\mathbb{E}[L(w_T) - L(w^*)]
    \\
    &\leq 
    \Bigg[ \prod_{s=1}^{T-1} (1 - 2\alpha\eta_s) \Bigg] \mathbb{E}[L(w_1)-L(w^*)] 
    +
    \sum_{t=1}^{T-1} \Bigg[ \prod_{s=t+1}^{T-1} (1 - 2\alpha\eta_s) \Bigg] \Bigg( \frac{G^2 \eta_t}{\sqrt{m_c}} + \beta G^2 \eta_t^2 \Bigg)
    .
\end{align*}
Now we set $\eta_t = 1/(2\alpha t)$ and observe that 
$\prod_{s=t+1}^{T-1} (1 - 2\alpha\eta_s) = \frac{t}{T-1},$
thus
\begin{align*}
    \mathbb{E}[L(w_T) - L(w^*)]
    \leq 
    \sum_{t=1}^{T-1} \frac{t}{T-1} \Bigg( \frac{G^2 \eta_t}{\sqrt{m_c}} + \beta G^2 \eta_t^2 \Bigg)
    \leq
    \frac{G^2}{2\alpha\sqrt{m_c}} + \frac{\beta G^2}{4\alpha^2} \frac{\log{T}}{T-1} 
    .&
\end{align*}
\end{proof}

\noindent\textbf{Additional observations regarding Theorem~\ref{theorem-convergnece}}: 
(i)~The theorem is stated with high probability because the selection of committees is random, and may thus potentially result in committees with a Byzantine majority. However, the probability of this bad event can be made arbitrarily low by choosing $N_p,N_c$ as in the theorem (see also \secref{sec:honest_majority});
(ii)~Controlling the number of samples $m_c$ used by the committee to evaluate and vote over the suggested updates at each round can allow us to reduce the impact of the relevant term in Theorem~\ref{theorem-convergnece} and converge closer to the optimal solution.

\section{Empirical Results}
\label{results}
We implemented our \textit{HoldOut SGD} algorithm as well as the methods described in Section~\ref{existing-rules}, and simulated Byzantine behavior using the attack described in Section~\ref{attack}. We report the results of existing methods and the \textit{HoldOut SGD} algorithm over the MNIST \cite{MNIST1998} and CIFAR-10 \cite{Krizhevsky2009LearningML} datasets.

\subsection{Implementation Details}
\subsubsection{\textbf{Models, Baselines and Data Generation}}
For MNIST, we use a fully connected network with 1 hidden layer, 784 dimensional flattened input, a 100-dimensional
hidden layer, and a 10-dimensional output, trained with cross-entropy loss objective and using ReLU activations.
The model was trained for 100 epochs with mini-batch size of 83 (as selected by \cite{Mhamdi2018TheHV, Baruch2019ALI}), a learning rate of 0.1 and no momentum.

For CIFAR-10, we use the LeNet-5 network \cite{LeCun1998GradientbasedLA} which is constructed from the following layers: (32x32x3) input, a convolutional layer
with kernel size: 5 $\times$ 5, 6 feature maps and a stride of 1, max-pooling layer
of size 2 $\times$ 2, a second convolutional layer with kernel 5 $\times$ 5, 16
feature maps and a stride of 1, a final max-pooling layer identical to the first one, followed by two fully connected layers of sizes 120 and 84 respectively, and an output layer of size 10. ReLU activations were used in all layers.
The model was trained for 1000 epochs with mini-batch size of 256, a learning rate of 0.1 and no momentum.

We follow the work of \cite{Baruch2019ALI} and \cite{Mhamdi2018TheHV} in the choice of the models, targeting simple architectures to demonstrate the effect of the Byzantine workers.

We compare \textit{HoldOut SGD} to two recent Byzantine-resilient aggregation rules, \textit{Krum} \cite{Blanchard2017MachineLW} and \textit{Coordinate-wise Trimmed Mean} \cite{Yin2018ByzantineRobustDL}. The current State-of-the-art is the \textit{Bulyan} method \cite{Mhamdi2018TheHV}, which is the evolution of both. It combines both aggregation algorithms and was shown in \cite{Baruch2019ALI} to suffer from similar vulnerabilities. Furthermore, \textit{Bulyan} provides theoretical guarantees to be Byzantine-resilient only up to $f < \frac{1}{4}$. Therefore, we chose to not include it in the empirical results of this paper. Finally, we also compare to the simplest aggregation of averaging the gradients (namely, standard SGD).

Data was generated as follows. A pool of 100 nodes was generated, where each node sampled $m$ examples from the dataset as its private internal dataset.
For MNIST we used $N_p = N_c = 30$ and an expected Byzantine rate of $f = 0.33$ and $m = 2000$.  For CIFAR-10 we used $N_p = N_c = 12$ and an expected Byzantine rate of $f = 0.33$ and $m = 1000$.
The choice of $f = 0.33$ is common practice, since it also represents the theoretical bounds of Byzantine-agreement algorithms.

\subsubsection{\textbf{Adversarial Attack over the Parameter Space} \label{attack}} The attack was identical for all Byzantine nodes selected as proposers, and was used over all aggregation methods and on both datasets. Byzantine committee members, relevant for the \textit{HoldOut SGD} algorithm only, were designed for a powerful attack, creating a coalition that votes for all Byzantine proposers first, followed by a remainder of random votes for honest proposers to reach the $N_p(1- f)$ vote count. 
 Note that this attack is a rather extreme case, which we do not expect in real-world settings. Thus, in typical scenarios, actual performance of \textit{HoldOut SGD} is expected to be much better than that reported here.

A general assumption in distributed learning is that the different datasets held by all honest workers are i.i.d. in nature and therefore can be represented by some normal distribution. Due to this distribution of the data, an adversary can estimate the mean and standard deviation of the reported updates from the last training epoch.  As the learning process converges, this estimation becomes more and more relevant to the distribution of the updates in the current epoch \cite{Bagdasaryan2018HowTB} allowing the attacker to calculate a \textit{margin of poisoning} \cite{Mhamdi2018TheHV}.
This margin represents the degree to which the attacker can perturb the estimated mean of honest parameter updates, while still remaining close enough to the mean so it is selected by the Byzantine-resilient aggregation rule. All of the existing attacks in parameter space rely on some concept in the spirit of the \textit{margin of poisoning}, aiming for the maximal undetected perturbation away from the mean. We focus on a similar attack strategy  \cite{Baruch2019ALI} where the adversary perturbs all parameters under the constraint of remaining as close as possible to the mean.

\begin{algorithm}[h]
   \caption{Attack over the Parameter Space: code for Byzantine worker $i$.}
   \label{alg:attack}
   
   \SetKwProg{getUpdate}{Function \emph{getUpdate}}{}{end}
   \getUpdate{($w_{t}$):}{
   Estimate $\mu_t\in\mathbb{R}^d$, $\sigma_t\in\mathbb{R}^d$, $\gamma_{max, t}\in\mathbb{R}$\;
    Set $\Vec{V_t} \leftarrow \mu_t + \gamma_{max, t} \sigma_t$\;
   \Return $\Vec{V_t}$
   }
\end{algorithm}

At each training epoch the adversary estimates $\mu_t$ and $\sigma_t$, the mean and standard deviation of the gradient updates reported by the honest workers $\{\nabla L_{i_1}(w_t), ..., \nabla L_{i_{N_p(1-f)}}(w_t)\}$.
This estimate can also be accomplished by a non-omniscient adversary, by controlling a small segment of corrupted nodes \cite{Baruch2019ALI}. In the experiments we performed, we assumed an omniscient adversary for the purpose of worst case analysis. The adversary then estimates $\gamma_{max, t}$ (see below) and uses it to calculate the final reported Byzantine update $\Vec{V_t}$ (see Algorithm~\ref{alg:attack}). The factor $\gamma_{max, t}$ is for the aggregation rule under attack and is chosen to be maximal, under the constraint of keeping the reported Byzantine update $\Vec{V_t}$ close enough to $\mu_t$ as possible to fool the aggregation rule into choosing $\Vec{V_t}$. Baruch et al. \cite{Baruch2019ALI} showed that $\gamma_{max, t} \approx 1.75$ is effective, where this value was derived from the CDF of the normal distribution.
In our empirical experiments we allow the omniscient attacker to evaluate  $\gamma_{max, t}$ at each round by running the aggregation rule internally and inflating $\gamma_{max, t}$ up to the point $\Vec{V_t}$ is no longer chosen by the aggregation rule.

\begin{figure}[t]
\vskip 0.2in
\begin{center}
\centerline{\includegraphics[width=\columnwidth]{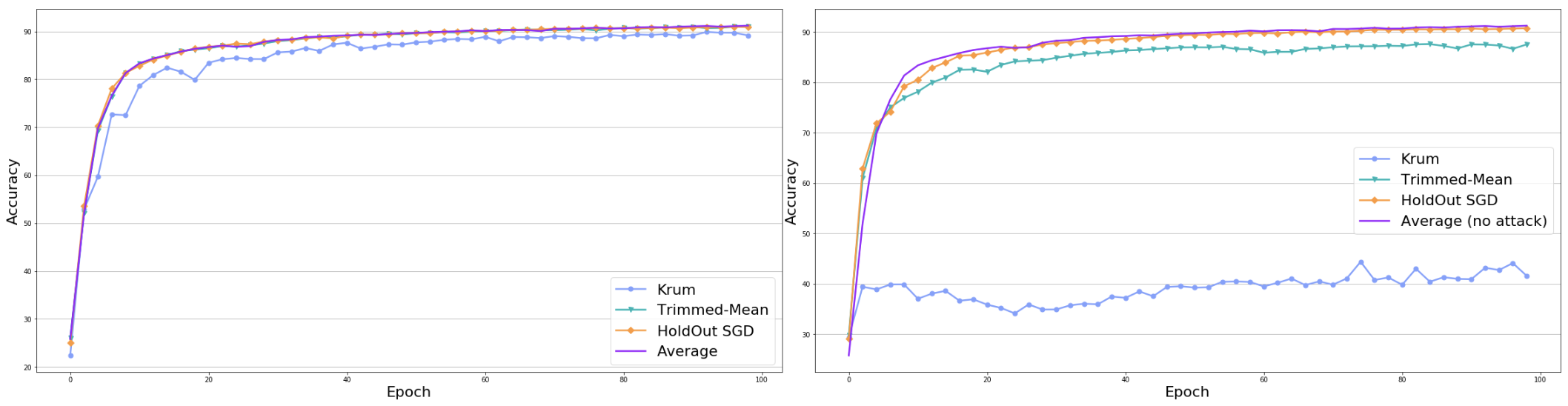}}
\caption{Accuracy during training for the MNIST dataset. Four aggregation rules are compared. {\bf Left}: A Byzantine-free setting. {\bf Right}: $0f=33\%$ Byzantines. The \textit{Average (no attack)} curve corresponds to the non Byzantine setting, and therefore an upper bound on accuracy.}
\label{mnist}
\end{center}
\vskip -0.2in
\end{figure}

\subsection{Results}
We first validated the convergence of the \textit{HoldOut SGD} algorithm by applying it to the MNIST dataset and comparing its convergence to other methods, in a Byzantine-free setting ($f = 0$).
Figure~\ref{mnist}(left) shows that \textit{HoldOut SGD} converges as fast as the average and \textit{Trimmed-Mean} aggregation rules, and all outperform \textit{Krum}.

Next, we evaluated learning under the attack of \secref{attack}. Results are shown in Figure~\ref{mnist}(right). As can be seen, except for \textit{HoldOut SGD}, the attack over the parameter space was enough to damage all other methods. \textit{Trimmed-Mean} suffered a 4\% decrease in test accuracy compared to accuracy without attack,  while \textit{Krum} converged into a completely ineffectual model suffering over 50\% decrease in test accuracy.

Finally, we ran the learning process over the CIFAR-10 dataset, with the Byzantine setting of \secref{attack}. Figure~\ref{cifar-10} shows that even though \textit{HoldOut SGD} suffers an impact under this crafted attack, it still outperforms all other methods. While \textit{Trimmed-Mean} suffered from  $~50\% $ decrease in test accuracy, compared to the average aggregation rule under no attack, and \textit{Krum} suffered over 65\%, \textit{HoldOut SGD} resulted in about 15\% decrease at the end of the learning process.

\begin{figure}[ht]
\vskip 0.2in
\begin{center}
\centerline{\includegraphics[width=0.8\columnwidth]{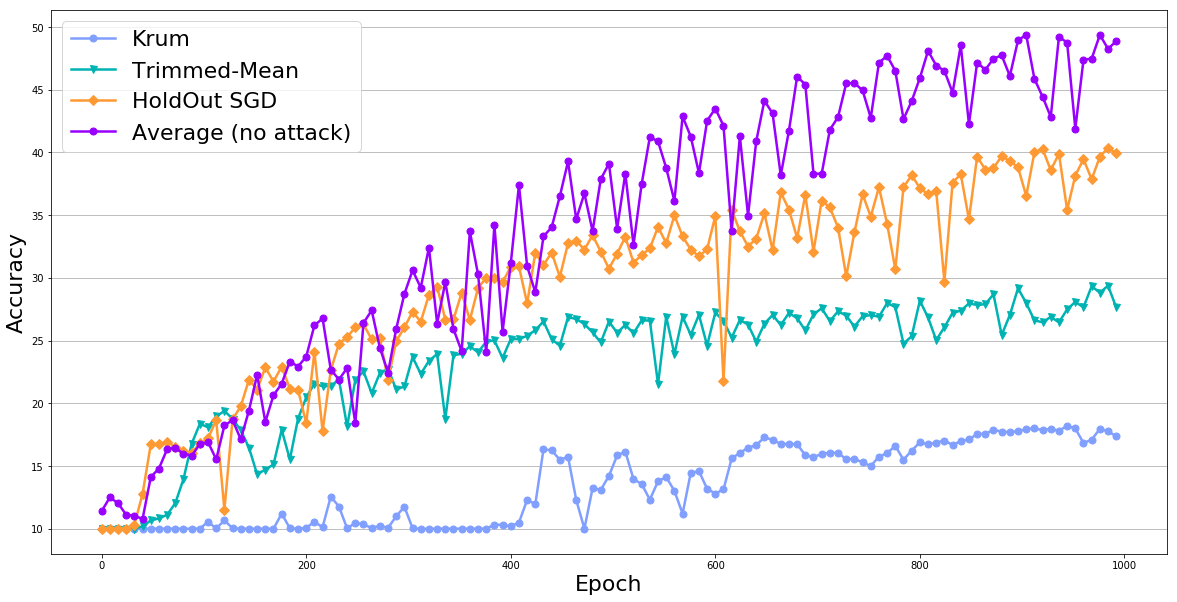}}
\caption{CIFAR-10: test accuracy up to training epoch 1000, comparing the performance of the different aggregation rules, under an actual Byzantine proportion of 33\%. The \textit{Average (no attack)} stands as reference to the performance under a Byzantine-free environment.}
\label{cifar-10}
\end{center}
\vskip -0.2in
\end{figure}

\section{Concluding remarks}
\label{conclusions}
Developing distributed SGD implementations that are
Byzantine-resilient and efficient is an important goal, given the growth in scale of user-data and the corresponding emerging privacy concerns.
A Byzantine tolerant SGD algorithm might be useful not only against a dedicated adversary whose goal is to generate an ineffectual model, but can also be useful in reducing the impact of biased data originating from unreliable sources.

Our theoretical results for Byzantine tolerance consider worst case attacks (e.g., including coalitions, omniscient Byzantine nodes having perfect estimation of gradients for honest workers etc.). Under the assumption that the target loss function is convex we provide convergence guarantees that tolerate up to half the workers being Byzantine ($f < \tfrac12$).

Our theoretical guarantees are for the convex case, since this is currently the main setting in which convergence to global optima can be proven. Obtaining global optimality results for non-convex settings is a challenging open question, and results are only available for rather limited settings such as linear networks, matrix completion, or very wide networks \cite{Arora2018ACA, Jacot2018NeuralTK, Ge2016MatrixCH}. That said, \textit{HoldOut SGD} is applicable to non-convex optimization and we show empirically that in these cases it is more robust to attacks than strong recent baselines.

\remove{
\textit{HoldOut SGD} borrows the concept of a committee voting on updates from Algorand \cite{Chen2019AlgorandAS} in order to achieve Byzantine tolerant selection of honest gradient updates, discarding the malicious updates. 
We leave for future work the analysis of the effects of different parameters (e.g., the size of the committee, the number of samples each committee member uses) on the accuracy of the model, and the time complexity of the algorithm.\agtodo{Didn't we actually do this analysis in Theorem 5... We can just discard this paragraph. Unless it is meant for the full decentralized method.}}

\textit{HoldOut SGD} was designed to tackle the Byzantine attacks aiming to cause the learning process to converge to a poor accuracy model. We did not cover the type of attacks referred to as ``backdooring'' where the adversary is trying to insert a hidden backdoor into the learned model, while preserving its performance over target test set.

Our fully decentralized implementation in \secref{sec:fully} offers the opportunity for fully distributed SGD algorithms that eliminate the central parameter server. This is motivated by the desire to have a more democratic learning algorithm, where there is no one entity that governs the algorithms that run on the central server. In the same way that federated blockchains take the governance power from one central entity (e.g., a bank) and give it to a group of entities.  Here, a group of entities will govern the learning process, thus turning the algorithm into a federated, federated learning (the first federated corresponds to the distribution of the governance, and the second for federated learning). We leave for future work a full evaluation of this approach, as well as further analysis of its computational and statistical properties.

\section{Acknowledgements}
This research is supported by the Blavatnik Computer Science Research Fund and by the European Research Council
(ERC) under the European Unions Horizon 2020 research
and innovation programme (grant ERC HOLI 819080). 
\bibliographystyle{unsrt}  
\bibliography{references}

\begin{thebibliography}{10}

\bibitem{McMahan2016CommunicationEfficientLO}
H.~Brendan McMahan, Eider Moore, Daniel Ramage, Seth Hampson, and
  Blaise~Ag{\"u}era y~Arcas.
\newblock Communication-efficient learning of deep networks from decentralized
  data.
\newblock In {\em AISTATS}, 2016.

\bibitem{Konecn2018FederatedLS}
Jakub Konecn{\'y}, H.~Brendan McMahan, Felix~X. Yu, Peter Richt{\'a}rik,
  Ananda~Theertha Suresh, and Dave Bacon.
\newblock Federated learning: Strategies for improving communication
  efficiency.
\newblock {\em ArXiv}, abs/1610.05492, 2018.

\bibitem{Kairouz2019AdvancesAO}
Peter Kairouz, H.~Brendan McMahan, Brendan Avent, Aur{\'e}lien Bellet, Mehdi
  Bennis, Arjun~Nitin Bhagoji, Keith Bonawitz, Zachary Charles, Graham Cormode,
  Rachel Cummings, Rafael G.~L. D'Oliveira, Salim~El Rouayheb, David Evans,
  Josh Gardner, Zachary~A. Garrett, Adri{\`a} Gasc{\'o}n, Badih Ghazi,
  Phillip~B. Gibbons, Marco Gruteser, Za{\"i}d Harchaoui, Chaoyang He, Lie He,
  Zhouyuan Huo, Ben Hutchinson, Justin Hsu, Martin Jaggi, Tara Javidi, Gauri
  Joshi, Mikhail Khodak, Jakub Konecn{\'y}, Aleksandra Korolova, Farinaz
  Koushanfar, Oluwasanmi Koyejo, Tancr{\`e}de Lepoint, Yang Liu, Prateek
  Mittal, Mehryar Mohri, Richard Nock, Ayfer {\"O}zg{\"u}r, Rasmus Pagh,
  Mariana Raykova, Hang Qi, Daniel Ramage, Ramesh Raskar, Dawn~Xiaodong Song,
  Weikang Song, Sebastian~U. Stich, Ziteng Sun, Ananda~Theertha Suresh, Florian
  Tram{\`e}r, Praneeth Vepakomma, Jianyu Wang, Li~Xiong, Zheng Xu, Qiang Yang,
  Felix~X. Yu, Han Yu, and Sen Zhao.
\newblock Advances and open problems in federated learning.
\newblock {\em ArXiv}, abs/1912.04977, 2019.

\bibitem{Li2014ScalingDM}
Mu~Li, David~G. Andersen, Jun~Woo Park, Alexander~J. Smola, Amr Ahmed, Vanja
  Josifovski, James Long, Eugene~J. Shekita, and Bor-Yiing Su.
\newblock Scaling distributed machine learning with the parameter server.
\newblock In {\em BigDataScience '14}, 2014.

\bibitem{Zinkevich2010ParallelizedSG}
Martin Zinkevich, Markus Weimer, Alexander~J. Smola, and Lihong Li.
\newblock Parallelized stochastic gradient descent.
\newblock In {\em NIPS}, 2010.

\bibitem{Feng2014DistributedRL}
Jiashi Feng, Huan Xu, and Shie Mannor.
\newblock Distributed robust learning.
\newblock {\em ArXiv}, abs/1409.5937, 2014.

\bibitem{Alistarh2018ByzantineSG}
Dan Alistarh, Zeyuan Allen-Zhu, and Jerry Li.
\newblock Byzantine stochastic gradient descent.
\newblock In {\em NeurIPS}, 2018.

\bibitem{Blanchard2017MachineLW}
Peva Blanchard, El~Mahdi~El Mhamdi, Rachid Guerraoui, and Julien Stainer.
\newblock Machine learning with adversaries: Byzantine tolerant gradient
  descent.
\newblock In {\em NIPS}, 2017.

\bibitem{Yin2018ByzantineRobustDL}
Dong Yin, Yudong Chen, Kannan Ramchandran, and Peter~L. Bartlett.
\newblock Byzantine-robust distributed learning: Towards optimal statistical
  rates.
\newblock In {\em ICML}, 2018.

\bibitem{Mhamdi2018TheHV}
El~Mahdi~El Mhamdi, Rachid Guerraoui, and S{\'e}bastien Rouault.
\newblock The hidden vulnerability of distributed learning in byzantium.
\newblock In {\em ICML}, 2018.

\bibitem{Chen2017DistributedSM}
Yudong Chen, Lili Su, and Jiaming Xu.
\newblock Distributed statistical machine learning in adversarial settings:
  Byzantine gradient descent.
\newblock In {\em SIGMETRICS 2017}, 2017.

\bibitem{Xie2018GeneralizedBS}
Cong Xie, Oluwasanmi Koyejo, and Indranil Gupta.
\newblock Generalized byzantine-tolerant sgd.
\newblock {\em ArXiv}, abs/1802.10116, 2018.

\bibitem{Su2016FaultTolerantMO}
Lili Su and Nitin~H. Vaidya.
\newblock Fault-tolerant multi-agent optimization: Optimal iterative
  distributed algorithms.
\newblock In {\em PODC '16}, 2016.

\bibitem{Chen2019AlgorandAS}
Jing Chen and Silvio Micali.
\newblock Algorand: A secure and efficient distributed ledger.
\newblock {\em Theor. Comput. Sci.}, pages 155--183, 2019.

\bibitem{Hastie2005TheEO}
Trevor~J. Hastie, Robert Tibshirani, and Jerome~H. Friedman.
\newblock The elements of statistical learning: Data mining, inference, and
  prediction, 2nd edition.
\newblock In {\em Springer Series in Statistics}, 2005.

\bibitem{Kohavi1995ASO}
Ron Kohavi.
\newblock A study of cross-validation and bootstrap for accuracy estimation and
  model selection.
\newblock In {\em IJCAI}, 1995.

\bibitem{BottouATStochasticGL}
Bottou.
\newblock Stochastic gradient learning in neural networks.
\newblock 1991.

\bibitem{Baruch2019ALI}
Moran Baruch, Gilad Baruch, and Yoav Goldberg.
\newblock A little is enough: Circumventing defenses for distributed learning.
\newblock {\em ArXiv}, abs/1902.06156, 2019.

\bibitem{Hazan2016IntroductionTO}
Elad Hazan.
\newblock Introduction to online convex optimization.
\newblock {\em ArXiv}, abs/1909.05207, 2016.

\bibitem{MNIST1998}
Yann LeCun.
\newblock The mnist database of handwritten digits.
\newblock 1998.

\bibitem{Krizhevsky2009LearningML}
Alex Krizhevsky.
\newblock Learning multiple layers of features from tiny images.
\newblock 2009.

\bibitem{LeCun1998GradientbasedLA}
Yann LeCun, L{\'e}on Bottou, Yoshua Bengio, and Patrick Haffner.
\newblock Gradient-based learning applied to document recognition.
\newblock 1998.

\bibitem{Bagdasaryan2018HowTB}
Eugene Bagdasaryan, Andreas Veit, Yiqing Hua, Deborah Estrin, and Vitaly
  Shmatikov.
\newblock How to backdoor federated learning.
\newblock {\em ArXiv}, abs/1807.00459, 2018.

\bibitem{Arora2018ACA}
Sanjeev Arora, Nadav Cohen, Noah Golowich, and Wei Hu.
\newblock A convergence analysis of gradient descent for deep linear neural
  networks.
\newblock {\em ArXiv}, abs/1810.02281, 2018.

\bibitem{Jacot2018NeuralTK}
Arthur Jacot, Franck Gabriel, and Cl{\'e}ment Hongler.
\newblock Neural tangent kernel: Convergence and generalization in neural
  networks.
\newblock {\em ArXiv}, abs/1806.07572, 2018.

\bibitem{Ge2016MatrixCH}
Rong Ge, Jason~D. Lee, and Tengyu Ma.
\newblock Matrix completion has no spurious local minimum.
\newblock In {\em NIPS}, 2016.

\end{thebibliography}

\end{document}